\documentclass[11pt]{article}

\usepackage{a4wide}
\usepackage{parskip}
\usepackage[utf8]{inputenc}

\usepackage{amsmath}
\usepackage{amssymb}

\usepackage{xspace}
\usepackage{tikz}

\usepackage{algorithm}
\usepackage[noend]{algpseudocode}

\usepackage[margin=1in]{geometry}
\usepackage{arydshln}
\usepackage{multirow}
\usepackage{wrapfig}
\usepackage{booktabs}
\usepackage{wrapfig}
\usepackage{caption}

\newtheorem{xdefinition}{Definition}[section]
\newtheorem{xobservation}[xdefinition]{Observation}
\newtheorem{xremark}[xdefinition]{Remark}
\newtheorem{xtheorem}[xdefinition]{Theorem}
\newtheorem{xlemma}[xdefinition]{Lemma}
\newtheorem{xproposition}[xdefinition]{Proposition}
\newtheorem{xcorollary}[xdefinition]{Corollary}
\newenvironment{definition}{\begin{xdefinition}\rm}%
{\hspace*{\fill}\raisebox{-1pt}{\boldmath$\Box$}\end{xdefinition}}
{\hspace*{\fill}\raisebox{-1pt}{\boldmath$\Box$}\end{xobservation}}
{\hspace*{\fill}\raisebox{-1pt}{\boldmath$\Box$}\end{xremark}}
\newenvironment{theorem}{\begin{xtheorem}\rm}{\end{xtheorem}}
\newenvironment{lemma}{\begin{xlemma}\rm}{\end{xlemma}}

\newenvironment{proof}{\begin{trivlist}\item[]{\bf Proof }}%
{\hspace*{\fill}\raisebox{-1pt}{\boldmath$\Box$}\end{trivlist}}

\usepackage{graphicx}
\usepackage{subcaption}
\usepackage{xspace} 
\usepackage{xcolor}
\usepackage{mathtools}
\usepackage{amsfonts} 

\usepackage{caption}
\usepackage{tikz}
\usetikzlibrary{decorations.pathreplacing}

\usepackage{xcolor} 
\usepackage{hyperref}

\newcommand{\opt}{\ensuremath{\operatorname{\textsc{Opt}}}\xspace} 
\newcommand{\A}{\ensuremath{\operatorname{\mathbb{A}}}\xspace}

\newcommand{\alg}{\A}
\newcommand{\Alg}{\A}
\newcommand{\BSA}{\ensuremath{\operatorname{\textsc{Bsa}}}\xspace}
\newcommand{\REALS}{\mathbb{R}}
\renewcommand{\epsilon}{\varepsilon}
\newcommand{\floor}[1]{\left\lfloor #1 \right\rfloor}
\newcommand{\ceil}[1]{\left\lceil #1 \right\rceil}
\newcommand{\eps}{\ensuremath{\varepsilon}\xspace}

\newcommand{\mls}{\ensuremath{m_{\text{LS}}}\xspace}
\newcommand{\mlsm}{\ensuremath{\rounddown{m_{\text{LS}}}}\xspace}
\newcommand{\ms}{\ensuremath{m_{\text{S}}}\xspace}

\newcommand{\mll}{\ensuremath{m_{\text{LL}}}\xspace}
\newcommand{\mllp}{\ensuremath{\roundup{m_{\text{LL}}}}\xspace}

\newcommand{\ngood}{\ensuremath{n_{\text{G}}}\xspace}
\newcommand{\mres}{\ensuremath{m_{\text{R}}}\xspace}

\newcommand{\sgood}{\ensuremath{s_{\text{G}}}\xspace}
\newcommand{\leveld}{\ensuremath{d}\xspace}
\newcommand{\leveldp}{\ensuremath{\roundup{\leveld}}\xspace}
\newcommand{\mlsb}{\ensuremath{m_{\text{LSB}}}\xspace}

\newcommand{\nb}{\ensuremath{n_{\text{B}}}\xspace}
\newcommand{\mblackm}{\ensuremath{\rounddown{m_{\text{B}}}}\xspace}
\newcommand{\mblack}{\ensuremath{m_{\text{B}}}\xspace}
\newcommand{\ml}{\ensuremath{m_{\text{W}}}\xspace}
\newcommand{\mlm}{\ensuremath{\rounddown{m_{\text{W}}}}\xspace}
\newcommand{\nblack}{\ensuremath{x_{\text{B}}}\xspace}
\newcommand{\sblack}{\ensuremath{s_{\text{B}}}\xspace}
\newcommand{\sblackp}{\ensuremath{\roundup{s_{\text{B}}}}\xspace}
\newcommand{\sblackm}{\ensuremath{\rounddown{s_{\text{B}}}}\xspace}
\newcommand{\eblack}{\ensuremath{e_{\text{B}}}\xspace}
\newcommand{\eblackm}{\ensuremath{\rounddown{e_{\text{B}}}}\xspace}
\newcommand{\sizel}{\ensuremath{S_{\text{W}}}\xspace}

\newcommand{\roundup}[1]{\ensuremath{#1^+}\xspace}
\newcommand{\rounddown}[1]{\ensuremath{#1^-}\xspace}

\definecolor{skgreen}{rgb}{0.7, 0.8, 0.7}

\newcommand{\templatebin}[5]{\filldraw [draw=black, fill=#1, text=black] (#2,#3) rectangle ++(25,#4) node[midway] {\footnotesize $#5$};}
\newcommand{\bb}[4]{\templatebin{blue!10!}{#1}{#2}{#3}{#4}}
\newcommand{\rb}[4]{\templatebin{red!10!}{#1}{#2}{#3}{#4}}

\newcommand{\dgb}[4]{\templatebin{black!10!}{#1}{#2}{#3}{#4}}
\newcommand{\dgbgood}[4]{\templatebin{black!30!}{#1}{#2}{#3}{#4}}
\newcommand{\lfb}[4]{\templatebin{green!10!}{#1}{#2}{#3}{#4}}
\newcommand{\lfbblack}[4]{\templatebin{skgreen}{#1}{#2}{#3}{#4}}

\title{Online Bin Covering with Advice\thanks{The first, second, and fourth
    authors were supported in part by the Danish Council for Independent
    Research, Natural Sciences, grant DFF-1323-00247. A preliminary version of
    this paper appeared in the 16th International Algorithms and Data
    Structures Symposium (WADS), volume 11646 of Lecture Notes in Computer
    Science, Springer 2019.}
}

\author{Joan Boyar   \\ University of Southern Denmark
\\ \textsf{joan@imada.sdu.dk}
  \and
  Lene M. Favrholdt  \\ University
  of Southern Denmark \\ \textsf{lenem@imada.sdu.dk}
  \and
  Shahin Kamali \\   University of Manitoba \\
  \textsf{shahin.kamali@umanitoba.ca}
\and
  Kim S. Larsen  \\ University
 of Southern Denmark \\ \textsf{kslarsen@imada.sdu.dk}}

\date{}

\begin{document}

\maketitle

\begin{abstract}
  The bin covering problem asks for covering a maximum number of bins
  with an online sequence of $n$ items of different sizes in the range
  $(0,1]$; a bin is said to be covered if it receives items of total
  size at least~$1$. We study this problem in the advice setting and
  provide asymptotically tight bounds of $\Theta(n \log \opt)$ on the size of advice required to achieve
  optimal solutions.

  Moreover, we show that any algorithm with advice
  of size $o(\log \log n)$ has a competitive ratio of at most~$0.5$. In
  other words, advice of size $o(\log \log n)$ is useless for
  improving the competitive ratio of~$0.5$, attainable by an online
  algorithm without advice. This result highlights a difference
  between the bin covering and the bin packing problems in the advice
  model: for the bin packing problem, there are several algorithms
  with advice of constant size that outperform online algorithms
  without advice. Furthermore, we show that advice of size
  $O(\log \log n)$ is sufficient to achieve an asymptotic competitive ratio
  of~$0.5\bar{3}$ which is strictly better than
  the best ratio~$0.5$ attainable by purely online algorithms. The
  technicalities involved in introducing and analyzing this algorithm
  are quite different from the existing results for the bin packing
  problem and confirm the different nature of these two
  problems.

  Finally, we show that a linear number of advice bits is
  necessary to achieve any competitive ratio better than $15/16$ for the
  online bin covering problem.
\end{abstract}

\section{Introduction}
In the bin covering problem~\cite{AssmannJKL84}, the input is a
multi-set of items of different sizes in the range $(0,1]$ which need
to be placed into a set of bins. A bin is said to be covered if the
total size of items in it is at least 1. The goal of the bin covering
problem is to place items into bins so that a maximum number of bins
is covered.  In the online setting, items form a sequence which is
revealed in a piece-by-piece manner; that is, at each given time, one
item of the sequence is revealed and an online algorithm has to place
the item into a bin without any information about the forthcoming
items. The decisions of the algorithm are irrevocable.

Bin covering is closely related to the classic bin packing problem and
is sometimes called the dual bin packing problem\footnote{%
  There is another problem, also sometimes referred to as ``dual bin
  packing'', which asks for maximizing the number of items packed into
  a fixed number of bins; for the advice complexity of that dual bin
  packing problem, see~\cite{Renault17,BorodinPS18}.}. The input to
both problems is the same.  In the bin packing problem, however, the
goal is to place items into a minimum number of bins so that the total
size of items in each bin is at most 1. Online algorithms for bin
packing can naturally be extended to bin covering. For example,
Next-Fit is a bin packing algorithm which keeps one ``open'' bin at
any time: To place an incoming item $x$, if the size of $x$ is smaller
than the remaining capacity of the open bin, $x$ is placed in the open
bin; otherwise, the bin is closed (never used again) and a new bin is
opened. Dual-Next-Fit~\cite{AssmannJKL84} is a bin covering algorithm
that behaves similarly, except that it closes the bin when the total
size of items in it becomes at least 1.

\subsection{Offline algorithms}
In the offline setting, the bin packing and bin covering problems are
NP-hard.  There is an asymptotic fully polynomial-time approximation
scheme (AFPTAS) for bin covering~\cite{JansenS03}. There are also bin
packing algorithms which open $\opt + o(\opt)$
bins~\cite{KarmarkarK82,R13,HR17}, where $\opt$ is the number
of bins in the optimal packing.  The additive term was improved
in~\cite{R13} and further, to $O(\log\opt)$, in~\cite{HR17}.

\subsection{Online algorithms}
Online algorithms are often compared under the framework of
competitive analysis. Roughly speaking, the \emph{competitive ratio}
of a bin covering (respectively bin packing) algorithm is the minimum
(respectively maximum) ratio between the number of bins covered
(respectively opened) by the algorithm and that of an optimal offline
algorithm on the same input.

Despite similarities between bin covering and bin packing, the status
of these problems are different in the online setting.  In the case of
bin covering, it is known that no online algorithm can achieve a
competitive ratio better than 0.5~\cite{CsirikT88}, while bin covering
algorithms such as Dual-Next-Fit~\cite{AssmannJKL84} have the best
possible competitive ratio of 0.5. Hence, we have a clear picture of
the complexity of deterministic bin covering under competitive
analysis. The situation is more complicated for the bin packing
problem. It is known that no deterministic algorithm can achieve a
competitive ratio of 1.54278~\cite{Balogh_et_al_lowerBound18} while
the best existing deterministic algorithm has a competitive ratio of
1.5783~\cite{balogh_et_al}.  Note there is a gap between the best
known upper and lower bounds.

\subsection{Online algorithms with advice}
Advice complexity is a formalized way of measuring how much knowledge
of the future is required for an online algorithm to obtain a certain
level of performance, as measured by the competitive ratio.  When such
advice is available, algorithms with advice could lead to semi-online
algorithms.  Unlike related approaches such as
``lookahead''~\cite{Grove95} (in which some forthcoming items are
revealed to the algorithm) and ``closed bin
packing''~\cite{AsgeirssonACEMPVWW02} (where the length of the input
is revealed), \emph{any} information can be encoded and sent to the
algorithm under the advice setting.  This generality means that lower
bound results under the advice model also imply strong lower bound
results on semi-online algorithms, where one can infer impossibility
results simply from the length of an encoding of the information a
semi-online algorithm is provided with.  Advice complexity is also
closely related to randomization; complexity bounds from advice
complexity can be transferred to the randomization case and vice
versa~\cite{BKKK17,BockenhauerHK14,Mikkelsen16,DurrKR16}.

We use the advice on tape model defined in~\cite{HKK10,BKKKM17}: The
advice is generated by a benevolent oracle with unlimited
computational power. The advice is written on a tape and the algorithm
knows its meaning.  This general approach has been studied for many
problems (we refer the reader to a recent survey on advice complexity
of online problems~\cite{BoyarFKLM17}). In particular, bin packing has
been studied under the advice
model~\cite{BoyarKLL16,RenaultRS15,ADKRR18}.

For {\em bin covering with advice}, the decision of where to pack the
$i$th item is based on the content of the advice tape and the sizes of
the first $i$ items. For any bin covering algorithm, \alg, and any input sequence, $\sigma$, $\alg(\sigma)$ and $\opt(\sigma)$ denote the number of bins covered by \alg and an optimal offline algorithm, \opt, respectively, when given the input sequence $\sigma$. A bin covering algorithm, \alg, is {\em
  $c$-competitive with advice complexity $s(n)$} if there exists a
constant $b$ such that, for all $n$ and for all input sequences
$\sigma$ of length at most $n$, there exists some advice $\Phi$ such
that $\alg(\sigma) \geq c \cdot \opt(\sigma) - b$ and at most $s(n)$ bits of
$\Phi$ are accessed by the algorithm.  If $c=1$ and $b=0$, the
algorithm is {\em optimal}.  For a given algorithm, \alg, with a given
advice complexity, $s(n)$, the {\em competitive ratio} is
$\sup\{c \mid \alg \text{is } c\text{-competitive}\}$. Thus, competitive
ratios are fractions, with $1$ being the best possible.
For the {\em asymptotic} competitive ratio, the $b$ term is allowed to
be non-constant, as long as it is $o(\opt)$. Thus, the
{\em asymptotic competitive ratio} of \alg is $\limsup_{\opt(\sigma)
  \rightarrow \infty} \{\alg(\sigma)/\opt(\sigma)\}$.

Note that bin covering is a maximization problem.  For minimization
problems, such as bin packing, the competitive ratio is defined
analogously, except that the inequality is
$\alg(\sigma) \leq c \cdot \opt(\sigma) + b$.  Similarly, the competitive
ratio is the {\em infimum} over all $c$ such that algorithm is
$c$-competitive, giving values at least $1$.

For bin packing, it is known that $\Theta(n \log(\opt))$ advice bits
are necessary and sufficient to produce optimal
solutions~\cite{BoyarKLL16}, but a constant number of advice bits are
sufficient to obtain a competitive ratio close to $1.47$, beating the
best possible online algorithm without advice~\cite{ADKRR18}.
Furthermore, $2n+o(n)$ advice bits suffice to get arbitrarily close to
a competitive ratio of $4/3$~\cite{BoyarKLL16}, and getting below
$1.17$ requires at least a linear number of bits~\cite{Mikkelsen16}.
In~\cite{RenaultRS15}, $(1+\epsilon)$-competitive online algorithms
using $O(n \cdot \frac{1}{\epsilon}\log\frac{1}{\epsilon})$ advice
bits are designed based on round and group techniques known from
offline algorithms.

\subsection{Contributions}
In this article, we provide the first results with respect to the
advice complexity of the bin covering problem.

To obtain an optimal result, advice essentially corresponding to an
encoding of an entire optimal solution is necessary and
sufficient. Not surprisingly, this follows from a similar proof for
bin packing, since for both problems, bins filled to size one in an
optimal solution are at the core of the proof.

However, unlike the bin packing problem, advice of constant size
cannot improve the competitive ratio of algorithms. We establish
this result by showing that any algorithm with advice of size
$o(\log \log n)$ has a competitive ratio of at most 0.5, which is the
optimal competitive ratio of online algorithms without advice.

We prove a tight result that advice of size $O(\log \log n)$ suffices
to achieve a competitive ratio arbitrarily close to $0.5\bar{3}$.
Some techniques that we develop for this result are quite different
from the existing results for bin packing and are likely helpful for
future analysis of bin covering with advice.  The idea is to let the
advice communicate the number of bins of certain ``types'' in an
optimal packing.  However, to get down to $O(\log\log n)$ bits of
advice, only approximate values are given. This idea is similar to
that for bin packing in~\cite{ADKRR18}, except that in~\cite{ADKRR18}
only a constant number of bits are used, approximating the ratio of
the number of bins in two different sets. 
Another difference between our algorithm and previous solutions for bin packing is that, we employ the ``dual-worst-fit'' scheme for placing certain items for bin covering. In contrast, the worst-fit algorithm for bin packing has not been used previously and does not seem to be helpful.

Finally, using a reduction from the binary string guessing
problem~\cite{BockenhauerHKKSS14}, we show that advice of linear size
is necessary to achieve any competitive ratio larger than 15/16. This
is similar to, but more intricate, than the corresponding result for
bin packing.

\subsection{Techniques}
We provide an intuitive explanation of the difference between bin
packing and bin covering under the advice model, and use this
explanation to describe our techniques in designing algorithms and impossibility results for the bin covering problem with advice.
Online bin packing is
relatively ``easy'' when items are relatively large (close to 1) or small (close to 0). Online
algorithms can place large and small items separately, and this gives
relatively good competitive ratios because an optimal algorithm has to
open a bin for every large item and the online algorithm can fill any
bin almost completely with small items. In contrast, in the bin packing
problem, the ``tricky'' items are those that are close to $1/2$ (a bit
more or less than $1/2$), and other items can be handled without
wasting too much space. For inputs formed by the tricky items, a bin
packing algorithm acts like a ``matching algorithm'', where items
smaller than $1/2$ can match with themselves or some items larger than
$1/2$. Advice can help by encoding the number of items slightly larger
than $1/2$. This is consistent with reserving some space for
``critical'' items, which is the main technique used in some results
for bin packing~\cite{BoyarKLL16,ADKRR18}. The resulting bins reflect
the matching of large and small items or small items among
themselves. It turns out that if we know the ratio between these two
``types'' of bins (approximated by a constant number of bits), we can
do better than any deterministic online algorithm (see~\cite{ADKRR18}
for details).  For bin covering, however, the tricky items are those
that are either very large (close to 1) or very small (close to
0). Inputs formed by such tricky items are used in~\cite{CsirikT88} to
derive a lower bound on the competitive ratio of purely online
algorithms. We also use them in Section~\ref{sect:lower} to derive
an impossibility result when the advice size is $o(\log \log n)$. Similarly to
bin packing, for inputs formed by tricky items, a bin covering
algorithm becomes a ``matching algorithm''. The matching process is
albeit a bit harder than it is for bin packing. This is because,
unlike bin packing where matching two large items is not possible (as
they do not fit in the bin), in bin covering, matching two large items
is possible and sometimes necessary.  To be a bit more precise, in bin
covering, we cannot afford having a bin with an unmatched large item;
such a bin will be almost full, but not covered, which implies that
the item placed in the bin is wasted.  When designing algorithms for
bin covering in Section~\ref{sec:loglog}, we exploit advice in order
to form packings that ensure that large items are always matched with
other items, preferably with a set of small items, and if not possible
with other large items. This ensures that items are not wasted in our
packings, that is, all bins that receive items will be covered (except
potentially $o(\opt)$ of them).
This family of packings is
formally defined in Section~\ref{sec:loglog} as
``$(\alpha,\eps)$-desirable packings'', in which large items are matched
with each other, except for a fraction $\alpha$ of them that are
placed in bins that are covered by small items. The value of $\eps$ is a
parameter determined by the accuracy of the encodings of
the ``approximate values'' of a few numbers passed
to the algorithm in the form of advice. 

In Section~\ref{sec:sublin}, we use a reduction from the binary
separation problem (closely related to the binary string guessing
problem) to derive an impossibility result when the advice size is
sub-linear. While reductions of this form have been used for the bin
packing problem, there are major differences between these
reductions. These differences are again rooted in the fact that hard
sequences for bin covering are formed by items that are close to 0 or
1, and that there are more possibilities for placing items in a bin
covering instance compared to bin packing.

Throughout the paper, we let the {\em level} of a bin denote the total
size of items packed in that bin at the current time.

\section{Optimal covering and advice}
It is not hard to see that advice of size $O(n \log(\opt(\sigma)))$ is
sufficient to achieve an optimal covering for an input $\sigma$ of
length $n$; note that $\opt(\sigma)$ denotes the number of bins in an
optimal covering of $\sigma$. Provided with $O(\log(\opt(\sigma)))$
bits of advice for each item, the offline oracle can indicate in which
bin the item is placed in the optimal packing. Provided with this
advice, the online algorithm just needs to pack each item in the bin
indicated by the advice. Clearly, the size of the advice is
$O(n \log(\opt(\sigma)))$ and the outcome is an optimal packing.  Note
that it is always assumed that the oracle that generates the advice
has unbounded computational power. However, if the time complexity of
the oracle is a concern, we can use the AFPTAS of~\cite{JansenS03} to
generate an almost-optimal packing and encode it in the advice.
Similarly, if the input is assumed to have only $m$ distinct known
sizes, one can encode the entire request sequence, specifying for each
distinct size how many of that size occur in the sequence. This only
requires $O(m\log(n))$ bits of advice.  The following theorem shows
that the above naive solutions are asymptotically tight.

\begin{theorem}
  For online bin covering on sequences $\sigma$ of length $n$, advice
  of size $\Theta(n \log \opt(\sigma))$ is required and sufficient to
  achieve an optimal solution, assuming
  $2\opt(\sigma) \leq (1-\epsilon) n$ for some positive value of
  $\epsilon$.  When the input is formed by $n$ items with $m\in o(n)$
  distinct, known item sizes, advice of size $\Theta(m \log n)$ is
  required and sufficient to achieve an optimal solution.
\end{theorem}
\begin{proof}
  The lower bounds follow immediately from the corresponding results
  for bin packing~\cite[Theorems~1,~3]{BoyarKLL16}. Since the optimal
  result in those proofs have all bins filled to level~$1$, any
  non-optimal bin packing would also lead to a non-optimal bin
  covering.
\end{proof}

\section{Advice of size $o(\log \log n)$ is not helpful}\label{sect:lower}
In this section, we show that advice of size $o(\log \log n)$ does not
help for improving the competitive ratio of bin covering
algorithms. This result is in contrast to bin packing where advice of
constant size can improve the competitive ratio. Our lower bound
sequence is similar to the one in~\cite{CsirikT88}, where the authors
proved a lower bound on the competitive ratio of purely online
algorithms.

\begin{theorem}
  There is no algorithm with advice of size $o(\log \log n)$ and
  competitive ratio better than $1/2$.
\end{theorem}

\begin{proof}
Consider a family of sequences formed as follows:
  $$\sigma_j = \langle \underbrace{\epsilon, \epsilon, \ldots,
  \epsilon}_{n \ \text{items}}, \underbrace{1-j \epsilon, 1-
  j\epsilon, \ldots, 1- j \epsilon}_{\lfloor n/j\rfloor
  \ \text{items}} \rangle$$
Here, $j$ ranges from 1 through $\lfloor\sqrt{n}\rfloor$, giving rise
to $\lfloor\sqrt{n}\rfloor$ sequences in the family. All sequences
start with the same prefix of $n$ items of size~$\epsilon$. We assume
that $\epsilon < \frac{1}{2n}$ to ensure that, even if all these items
are placed in the same bin, the level of that bin is still less than
1/2. Note that the suffix, formed by items of size $1-j\epsilon$, has
length $O(n)$, so the length of any of the sequences is $\Theta(n)$.

Clearly, for packing $\sigma_j$, an optimal algorithm places $j$ items
of size $\epsilon$ in each bin and covers $\lfloor n/j\rfloor$
bins. So we have $\opt(\sigma_j) = \lfloor n/j\rfloor$.

The proof is by contradiction, so assume there is an algorithm, \alg,
using $o(\log\log n)$ advice bits and having competitive ratio
$1/2+\mu$ for some constant $0<\mu\leq\frac12$. Then, there exists a fixed
constant $d$ such that for any sequence $\sigma_j$, we have
\begin{equation}
  \alg(\sigma_j) \geq (1/2+\mu) \opt(\sigma_j) - d =
  \frac{1}{2}\left\lfloor\frac{n}{j}\right\rfloor + \mu
  \left\lfloor\frac{n}{j}\right\rfloor - d
\label{eq:lowbnd1}
\end{equation}

We say two sequences belong to the same \emph{sub-family} if they
receive the same advice string.  Since the advice has size $o(\log\log
n)$, there are $o(\log n)$ sub-families.  Let $\sigma_{a_1}, \ldots,
\sigma_{a_w}$ be the sequences in one sub-family.  Since the advice
and the first $n$ items (of size $\epsilon$) are the same for any two
members of this sub-family, \Alg will place these $n$ items
identically. Let $m_i$ denote the number of bins receiving at least
$i$ items in such a placement. So, we have $\sum_{i=1}^n m_i = n$ (a
bin with $x$ items is counted $x$ times, since it receives at least
one item, at least two items, etc.). Moreover, for any $\sigma_j$, we
have
\begin{equation}
  \alg(\sigma_j) \leq m_j + (\lfloor n/j\rfloor-m_j)/2 = \frac{1}{2}\left\lfloor\frac{n}{j}\right\rfloor
  + \frac{m_j}{2}
\label{eq:lowbnd2}
\end{equation}

This follows since any bin with at least $j$ items of size $\epsilon$
can be covered using only one item of size $1-j\epsilon$, while the
other bins require two such items.  From Equations~(\ref{eq:lowbnd1})
and~(\ref{eq:lowbnd2}), we get $\mu
\left\lfloor\frac{n}{j}\right\rfloor \leq \frac{m_j}{2} + d$. Summing
over $j \in \{ a_1, \ldots, a_w\}$, we get that
$$\mu n \left( \frac{1}{a_1}+\frac{1}{a_2} + \ldots + \frac{1}{a_w} \right)-\mu w \leq \frac12(m_{a_1} + m_{a_2} + \ldots + m_{a_w}) + wd,$$
where $\mu w$ is subtracted to compensate for removing the floor.
Since $\frac12 (m_{a_1} + m_{a_2} + \ldots + m_{a_w}) + (d+\mu)w \leq \frac12 \cdot \sum\limits_{i=1}^n m_i + (d+\frac12)n = (d+1)n$, we have
$$\frac{1}{a_1}+\frac{1}{a_2} + \ldots + \frac{1}{a_w} \in O(1)$$
Summing the left-hand side over all families, we include every
sequence and therefore every fraction, $\frac{1}{i}$, once and obtain
$\Sigma_{i=1}^{\lfloor\sqrt{n}\rfloor}\frac{1}{i}$. Since there are
$o(\log n)$ sub-families, and the contribution to the sum of fractions
from each sub-family has been proven constant, it follows that
$\Sigma_{i=1}^{\lfloor\sqrt{n}\rfloor}\frac{1}{i} \in o(\log n)$.
This is a contradiction, since the Harmonic number
$\Sigma_{i=1}^{\lfloor\sqrt{n}\rfloor}\frac{1}{i}\in\Theta(\log
\sqrt{n})=\Theta(\log n)$.

Thus, our initial assumption is wrong and with advice of size
$o(\log \log n)$, no algorithm with competitive ratio strictly better
than $1/2$ can exist.

For clarity in the exposition, we have ignored the issue of the
sequences in the family having different lengths, all of which are
larger than~$n$.  However, all sequences have length at most~$2n$, so
we can take that to be our~$n$, and pad all sequences with items small
enough that even the sum of them cannot fill up the missing space in
any bin. For the asymptotic result proven here, changing $n$ by a
factor of at most~$2$ is immaterial.
\end{proof}

\section{An algorithm with advice of size $O(\log \log n)$}\label{sec:loglog}
In this section, we show that advice of size $O(\log \log n)$ is
sufficient to achieve an asymptotic competitive ratio of
$8/15 = 0.5\bar{3}$.

Throughout this section, we call an item \emph{small}
if it has size less than 1/2 and \emph{large} otherwise.

\newcommand{\MYFORMATA}[1]{\begin{minipage}{1cm}\mbox{}\vspace*{.3ex}\par\centerline{#1}\vspace*{-1ex}\mbox{}\end{minipage}}
\newcommand{\MYFORMATB}[1]{\begin{minipage}{9.5cm}\mbox{}\vspace*{.3ex}\par #1\newline\vspace*{-1ex}\mbox{}\end{minipage}}

\begin{table}[!b]
\caption{Notation used in Section~\ref{sec:loglog}}\label{table:notation}
\begin{center}
\begin{tabular}[!b]{ |c|c| } 
 \hline 
 \MYFORMATA{Notation} & \MYFORMATB{Meaning}  \\ 
 \hline \hline
 \MYFORMATA{$n$} & \MYFORMATB{The length of the input.} \\  \hline
 \MYFORMATA{$\mls$} & \MYFORMATB{The number of LS-bins in the optimal packing.} \\ \hline
 \MYFORMATA{$\mll$} & \MYFORMATB{The number of LL-bins in the optimal packing.} \\ \hline
 \MYFORMATA{$\ms$} & \MYFORMATB{The number of S-bins in the optimal packing.} \\ \hline
 \MYFORMATA{$\beta$} & \MYFORMATB{The value of $\frac{\mls+\mll}{\mls}$. The algorithm behaves differently when  $\beta \geq 15/14$ compared to when $\beta < 15/14$.} \\ \hline
  \MYFORMATA{$\alpha$} & \MYFORMATB{A parameter of the algorithm when $\beta < 15/14$. Approximately $\lfloor \alpha \mls \rfloor$ of covered bins include exactly one large item.} \\   \hline
  \MYFORMATA{\eps} & \MYFORMATB{$\eps \in O\left(\frac{1}{\log n}\right)$.} \\
 \hline
  \MYFORMATA{\ngood} & \MYFORMATB{$\ngood = \floor{\alpha \mls/3}$ is
    the number of good items.} \\
   \hline
  \MYFORMATA{\leveld} & \MYFORMATB{$\leveld = 1 - \sgood$, where \sgood is the size of the smallest good item.} \\
   \hline
  \MYFORMATA{\sblack} & \MYFORMATB{The size of the \nb'th smallest
    black item.} \\
   \hline
  \MYFORMATA{\mres} & \MYFORMATB{The number of reserved bins.} \\
   \hline
  \MYFORMATA{\mblack} & \MYFORMATB{Among the reserved bins, \mblack
    bins are black bins.} \\
 \hline
  \MYFORMATA{\ml} & \MYFORMATB{Among the reserved bins, \ml
    bins are white bins.}\\
 \hline
\end{tabular}
\end{center}
\end{table}

Consider a packing of the input sequence $\sigma$. We partition the
bins in this packing into three groups. A \emph{large-small} (LS) bin
includes one large item and some small items, a \emph{large-large}
(LL) bin includes only two large items, and a \emph{small} (S) bin
includes only small items.  We use $\ms$, $\mls$, and $\mll$,
respectively, to denote the number of S-bins, LS-bins, and LL-bins in
the optimal packing. For $\mls \geq 1$, we let $\beta\geq 1$ satisfy
$\mls+\mll = \beta \mls$. See Table~\ref{table:notation} for a summary
of notation used in this section.

\subsection{A simple algorithm for the case $\mls = 0$ or $\beta\geq 15/14$}
The following lemma shows that sequences, where there are not too many LS-bins compared to LL-bins
in an optimal packing, are ``easy'' instances. More specifically, when
there are at most $14$ times as many LS-bins as LL-bins, there is a
simple $8/15$-competitive algorithm.
\begin{lemma}\label{lem:easycaselog}
  When $\mls=0$ or $\beta \geq 15/14$,
  there is an online bin covering algorithm with competitive ratio at least $8/15$.
\end{lemma}
\begin{proof}
  Consider a simple algorithm, \alg, that places large and small items
  separately. Each pair of large items cover one bin and small items
  are placed using the Dual-Next-Fit strategy, that is, they are
  placed in the same bin until the bin is covered (and then a new bin
  is started).  Let $S$ denote the total size of small items. Note
  that the number of large items is $\mls+2\mll$. The number of bins
  covered by \alg is at least
  $\lfloor (\mls+2\mll)/2 \rfloor + \lfloor 2S/3 \rfloor
\geq (\mls+2\mll)/2  +  2S/3 -2 $.  The number of
  bins covered by \opt is at most $\mls+\mll+\lfloor S \rfloor$.  Thus, for
  any input sequence, $\sigma$,
\begin{align*}
 \alg(\sigma)
 & \geq \frac{(\mls+2\mll)/2 + 2S/3}{\mls+\mll+S} \cdot \opt(\sigma) - 2
\end{align*}
For $\mls=0$, this proves a ratio of at least $2/3$ which is larger
than $8/15$. For $\beta
\geq 15/14$, this gives a competitive ratio of at least
$\min\{2/3,\frac{2\beta-1}{2\beta}\}$, since
\begin{align*}
  \frac{(\mls+2\mll)/2}{\mls+\mll}
    = \frac{2\mls+2\mll-\mls}{2\mls+2\mll}
    = \frac{2 \frac{\mls+\mll}{\mls}-1}{2 \frac{\mls+\mll}{\mls}}
    = \frac{2\beta-1}{2\beta}.
\end{align*}
For $\beta \geq 15/14$, $(2\beta-1)/(2\beta) \geq 8/15$.
This proves the lemma.
\end{proof}

\subsection{A more complicated algorithm for the case $\mls \geq 1$ and $\beta < 15/14$}
In what follows, we define $(\alpha,\eps)$-desirable packings, which
act as reference packings for our algorithm. Here $\alpha$ and $\eps$
are parameters of the algorithm that we will introduce later.  Note
that in an optimal packing, there are $\mls$ bins covered by one large
item and some small items. Consider a covering in which large and
small items are placed in separate bins. Clearly, there are 0 bins
that are covered by one large and some small items and, in the worst
case, the level of all bins covered by small bins is almost 1.5 (when
all small items have size a bit less than 0.5).  Intuitively speaking,
an $(\alpha,\eps)$-desirable packing is an intermediate solution
between this packing and the optimal packing, in which there are
almost $\alpha \mls$ bins covered by one large item and some small
items.

For the following definition, it may be helpful to confer with
Figure~\ref{fig:optvsdes}.

\begin{definition}
\label{def:desirable}
  A packing is $(\alpha,\eps)$-\emph{desirable}, where $\alpha$ is a
  real number in the range $(0,1]$ and $\eps$ is a (small) positive real, if
  and only if all the following hold:
\begin{itemize}
\item[I] The packing has at least $\lfloor \alpha \mls \rfloor$ LS-bins.
  All LS-bins, except possibly $O(\eps\mls)$ of them, are
  covered.
\item[II] The large items that are not packed in LS-bins appear in
  pairs, with each pair covering one bin (except one item when there
  are an odd number of such large items).
\item[III] The packing has at
  least $\frac{2\ms}{3} - O(\eps\mls)$ covered S-bins.
\end{itemize}  
\end{definition}

\begin{lemma}\label{lem:alphades}
  For any input sequence, $\sigma$, the number of bins covered in an
  $(\alpha,\eps)$-desirable packing is at least
  $ \min \{ \frac{\alpha + 2\beta -1}{2\beta} ,
  \frac{2}{3} \} \cdot \opt(\sigma) - O(\eps\mls)$.
\end{lemma}

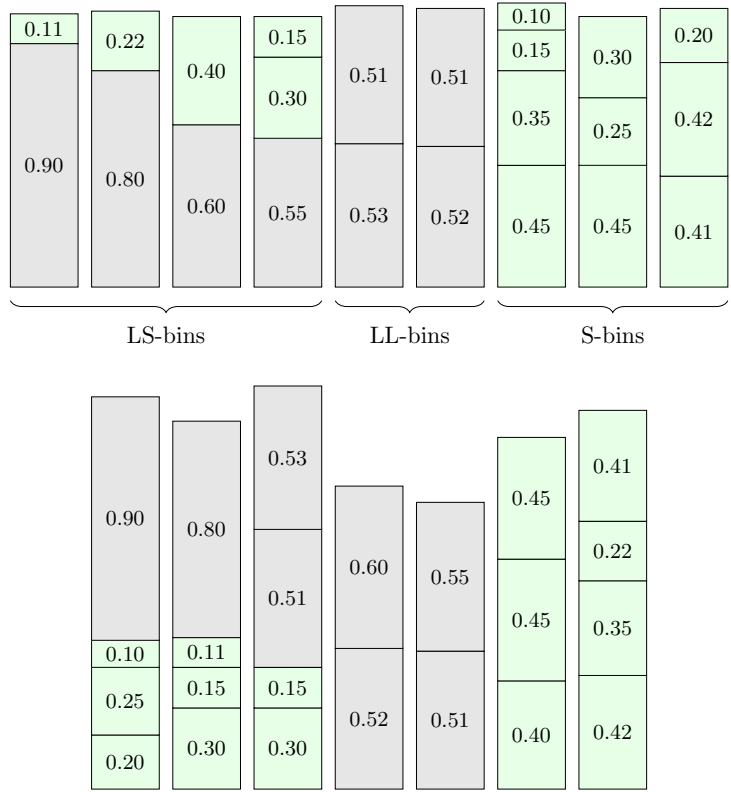
\begin{figure}[!t]
  \begin{center}
\scalebox{0.9}{
\begin{tikzpicture}[scale=0.04]
  \dgb{  0}{ 0}{90}{0.90}
  \lfb{  0}{90}{11}{0.11}
  \dgb{ 30}{ 0}{80}{0.80}
  \lfb{ 30}{80}{22}{0.22}
  \dgb{ 60}{ 0}{60}{0.60}
  \lfb{ 60}{60}{40}{0.40}
  \dgb{ 90}{ 0}{55}{0.55}
  \lfb{ 90}{55}{30}{0.30}
  \lfb{ 90}{85}{15}{0.15}
  \dgb{120}{ 0}{53}{0.53}
  \dgb{120}{53}{51}{0.51}
  \dgb{150}{ 0}{52}{0.52}
  \dgb{150}{52}{51}{0.51}
  \lfb{180}{ 0}{45}{0.45}
  \lfb{180}{45}{35}{0.35}
  \lfb{180}{80}{15}{0.15}
  \lfb{180}{95}{10}{0.10}
  \lfb{210}{ 0}{45}{0.45}
  \lfb{210}{45}{25}{0.25}
  \lfb{210}{70}{30}{0.30}
  \lfb{240}{ 0}{41}{0.41}
  \lfb{240}{41}{42}{0.42}
  \lfb{240}{83}{20}{0.20}
  \draw [decorate,decoration={brace,amplitude=5pt,mirror},xshift=0pt,yshift=0pt] (0,-5) -- (115,-5) node [midway,yshift=-.5cm] {\small LS-bins};
  \draw [decorate,decoration={brace,amplitude=5pt,mirror},xshift=0pt,yshift=0pt] (120,-5) -- (175,-5) node [midway,yshift=-.5cm] {\small LL-bins};
  \draw [decorate,decoration={brace,amplitude=5pt,mirror},xshift=0pt,yshift=0pt] (180,-5) -- (265,-5) node [midway,yshift=-.5cm] {\small S-bins};
\end{tikzpicture}
}

\bigskip

\raisebox{1.58ex}{
\scalebox{0.9}{
\begin{tikzpicture}[scale=0.04]
  \lfb{  0}{ 0}{20}{0.20}
  \lfb{  0}{20}{25}{0.25}
  \lfb{  0}{45}{10}{0.10}
  \dgb{  0}{55}{90}{0.90}
  \lfb{ 30}{ 0}{30}{0.30}
  \lfb{ 30}{30}{15}{0.15}
  \lfb{ 30}{45}{11}{0.11}
  \dgb{ 30}{56}{80}{0.80}
  \lfb{ 60}{ 0}{30}{0.30}
  \lfb{ 60}{30}{15}{0.15}
  \dgb{ 60}{45}{51}{0.51}
  \dgb{ 60}{96}{53}{0.53}
  \dgb{ 90}{ 0}{52}{0.52}
  \dgb{ 90}{52}{60}{0.60}
  \dgb{120}{ 0}{51}{0.51}
  \dgb{120}{51}{55}{0.55}
  \lfb{150}{ 0}{40}{0.40}
  \lfb{150}{40}{45}{0.45}
  \lfb{150}{85}{45}{0.45}
  \lfb{180}{ 0}{42}{0.42}
  \lfb{180}{42}{35}{0.35}
  \lfb{180}{77}{22}{0.22}
  \lfb{180}{99}{41}{0.41}
\end{tikzpicture}
}
}
\caption{(top) an optimal packing with
  $\mls=4, \mll=2$, and $\ms= 3$;
  (bottom) an $(\alpha,\eps)$-desirable packing with $\alpha = 1/2$ and
  $\eps = 1/64$.}
\label{fig:optvsdes}
\end{center}
\end{figure}

\begin{proof}
  The number of bins covered in the optimal packing is
  $\opt(\sigma) = \mll+ \mls + \ms = \beta \mls + \ms$.

  Now consider the $(\alpha,\eps)$-desirable packing and
  let $x_{\text{LS}}$ be the number of LS-bins in the packing.
  Then by Definition~\ref{def:desirable},
    $x_{\text{LS}} \geq \lfloor \alpha \mls \rfloor$,
  and the number of covered LS-bins is $x_{\text{LS}} - O(\eps\mls)$.
  The number of covered LL-bins is $\lfloor (2\mll+\mls-x_{\text{LS}})/2
  \rfloor = \lfloor ((2\beta-1)\mls-x_{\text{LS}})/2 \rfloor$.
  By Definition~\ref{def:desirable}, the number of covered S-bins is
  at least $\frac{2\ms}{3} - O(\eps\mls)$.
  Hence, the number of bins covered in the
  $(\alpha,\eps)$-desirable packing of $\sigma$ will be at least

\begin{align*}
&    x_{\text{LS}} - O(\eps\mls)
     + \left\lfloor \frac{(2\beta-1)\mls-x_{\text{LS}}}{2} \right\rfloor
     + \frac{2\ms}{3} - O(\eps\mls)\\
& =    \frac12 x_{\text{LS}}
     + \frac{(2\beta-1)\mls}{2}
     + \frac{2\ms}{3}
     - O(\eps\mls)  \\
& \geq \frac{\alpha \mls}{2}
     + \frac{(2\beta-1)\mls}{2}
     + \frac{2\ms}{3}
     - O(\eps\mls)  \\
& = \frac{\alpha \mls / 2 + (2\beta-1)\mls/2 + 2\ms/3}{\beta \mls + \ms} \cdot \opt(\sigma) - O(\eps\mls) \\
& \geq \min\left\{\frac{\alpha+2\beta-1}{2\beta}, \frac{2}{3} \right\} \cdot \opt(\sigma) - O(\eps\mls).
\end{align*}

In the first equality, we simply remove terms that are anyway absorbed
in $O(\eps\mls)$; similar for the first inequality, using the bound
$x_{\text{LS}} \geq \lfloor \alpha \mls \rfloor$, where the constant
coming from lifting the floor is also absorbed.  In the last equality,
we have gathered the first three terms and multiplied with the
identity $\opt(\sigma)/(\beta \mls + \ms)$.  Finally, for the last
inequality, we have used the arithmetic inequality,
$\frac{a+b}{c+d}\geq\min(\frac{a}{c},\frac{b}{d})$, to collect the
$\mls$ and $\ms$ terms, respectively.
\end{proof}

In the remainder of this section, we describe an algorithm that
achieves an $(\alpha,\eps)$-desirable packing for certain values of
$\alpha$ and $\eps$.  Here, $\eps$ is a parameter determined by
the accuracy of the encodings of
approximate values of a few numbers passed to the algorithm. Before
describing these numbers, we explain how the approximate encodings
work.

Our algorithm requires approximate encoding of integers (associated
with the number of certain bins), as well as positive values smaller
than~$1$ (associated with size of items or levels of bins).  Given a
positive integer $x$, we can write the length of the binary encoding
of $x$ in $O(\log \log x)$ bits, using self-delimited encoding as
in~\cite{Komm16}.  The approximate value of $x$ will be represented by
the binary encoding of the length of $x$, plus the $k$ most
significant bits of $x$ after the high-order $1$, with $k =
\ceil{\log\log x}$. Thus we use $O(\log\log n)$ bits in total for
encoding numbers no larger than $n$.  Setting the unknown lower order
bits to zero gives an approximation to $x$ which we denote by
$\rounddown{x}$. We can bound $\rounddown{x}$ as follows: If
$\rounddown{x}=y \cdot 2^\ell$ for some $y$ represented by $k+1$ bits,
where the high-order bit is a one, then $2^k\leq y < 2^{k+1}$. Given
$\rounddown{x}$, the largest $x$ could be is $y \cdot 2^\ell + (2^\ell
-1)$. Thus, $(1-\frac{1}{2^k})x <\rounddown{x} \leq x$.  Similarly,
setting the unknown lower order bits to one gives an approximation to
$x$ which we denote by $\roundup{x}$. We can bound $\roundup{x}$ as
follows: $x \leq \roundup{x} < (1+\frac{1}{2^k})x$.  For a
real-valued, positive $x<1$, we consider its binary representation as
the sum of powers of $1/2$.  To represent $\rounddown{x}$, we write
the first $k = O(\log \log n)$ powers of $1/2$, ignoring any unknown
remaining bits. If these unknown bits are all~$1$, we get
$\roundup{x}$, the value of which is no more than $x+\frac{1}{2^k}$.

We will let $\eps = \frac{1}{2^k} \leq \frac{1}{\log n}$.

In the remaining more technical part of the section, it may be
beneficial to consider that if we had had $O(\log n)$ bits of advice
instead of $O(\log\log n)$, many arguments would be simplified, and it
could be helpful on a first reading to ignore multiplicative terms
such as $1-\eps$ that are there because we know only
approximate as opposed to exact values of the parameters we receive
information about in the advice.

\subsubsection{Outlining the algorithm} 

We call the $\ngood = \floor{\mls/3}$ largest items in the
input sequence \emph{good items} and
let $d = 1 - \sgood$, where \sgood denotes the size of the smallest
good item.  Thus, packing a good item in a bin with level at least
\leveld ensures that the bin is covered.

Our algorithm aims to place a subset of small items into \ngood \emph{reserved bins} in a way that the level of each reserved bin is at least $\leveld$, while the remaining small items will cover almost $2\ms/3$ separate bins. As we will describe later, some of the reserved bins will receive an additional good item and hence become covered. 

At first glance, placing small items to provide the above guarantee does not seem too hard: an optimal packing covers $\ms$ small bins while we need to cover roughly $2\ms/3$ bins, and the remaining small items have total size at least $2\mls/3 \times \leveld$ (because $2/3$ of \opt's LS-bins
have large items of size at most $1-\leveld$) while we want to cover a level $ \ngood \times \leveld \approx (\mls/3) \times \leveld$ (in a sense, the total volume of available small items is at least twice the required volume for achieving the desired level). Ideally, if we can balance the level of small items in the reserved bins, then all levels are close to the average level (roughly $2\leveld$). This can be done via the Dual-Worst-Fit algorithm that places each item in the bin that currently has the minimum level among a predefined number of bins (here \ngood, the number of reserved bins). 
Unfortunately however, the discrepancy between the sizes of the small items can create situations where simply using Dual-Worst-Fit is not sufficient to achieve the desired level for all reserved bins. Consequently, the algorithm is more delicate as we will describe in what follows.

The algorithm receives the approximation \leveldp as advice.
Small items of size at least \leveldp are called {\em black items}, and
items smaller than \leveldp are called {\em white items}.
Let \mlsb denote the number of LS-bins in the optimal packing that
contain a black item.

The following offline scheme illustrates the idea behind our algorithm.
See Figure~\ref{fig:optvsdes2} for an illustration.
\begin{itemize}
\item Open \ngood bins called {\em reserved} bins.
\item Place the $ \nb = \min\{\mlsb,\ngood\}$ smallest black
  items in \nb of these bins, one item per bin.
  These bins are called {\em black bins}.
\item Place white items of total size at least \leveld, but less than
  $\leveld + \leveldp$, in each of the
  remaining $\ml = \ngood - \nb$ reserved bins.
  These bins are called {\em white bins}.
  This is done using Dual-Worst-Fit. 
\item  The remaining small items are packed in S-bins using Dual-Next-Fit.
\item Place a large item in each of the reserved bins, ensuring that at
  least $\alpha \mls$ of them are good items.
  Place an additional large item in each reserved bin without a good item.
  Place the remaining large items pairwise in new bins.
\end{itemize}

\begin{figure}[!t]
  \begin{center}
\scalebox{0.76}{
\begin{tikzpicture}[scale=0.04]
  \dgbgood{  0}{ 0}{90}{0.90}
  \lfb{  0}{90}{10}{0.10}
  \dgbgood{ 30}{ 0}{85}{0.85}
  \lfb{ 30}{85}{16}{0.16}
  \dgbgood{ 60}{ 0}{78}{0.78}
  \lfb{ 60}{78}{14}{0.14}
  \lfb{ 60}{92}{9}{0.09}
  \dgbgood{ 90}{ 0}{77}{0.77}
  \lfb{ 90}{77}{12}{0.12}
  \lfb{ 90}{89}{11}{0.11}
  
   \dgb{ 120}{ 0}{75}{0.75}
  \lfb{ 120}{75}{16}{0.16}
  \lfb{ 120}{91}{8}{0.08}
  \dgb{ 150}{ 0}{74}{0.74}
  \lfbblack{ 150}{74}{26}{0.26}
  \dgb{ 180}{ 0}{73}{0.73}
  \lfb{ 180}{73}{20}{0.20}
  \lfb{ 180}{93}{8}{0.08}
  \dgb{ 210}{ 0}{71}{0.71}
  \lfb{ 210}{71}{23}{0.23}
  \lfb{ 210}{94}{7}{0.07}
  \dgb{ 240}{ 0}{68}{0.68}
  \lfb{ 240}{68}{15}{0.15}
  \lfb{ 240}{83}{10}{0.10}
  \lfb{ 240}{93}{7}{0.07}
  \dgb{ 270}{ 0}{65}{0.65}
  \lfbblack{ 270}{65}{24}{0.24}
  \lfb{ 270}{89}{12}{0.12}
  \dgb{ 300}{ 0}{64}{0.64}
  \lfb{ 300}{64}{22}{0.22}
  \lfb{ 300}{86}{7}{0.07}
  \lfb{ 300}{93}{7}{0.07}
  \dgb{ 330}{ 0}{63}{0.63}
  \lfb{ 330}{ 63}{23}{0.23}
  \lfb{ 330}{ 86}{15}{0.15}
	
\lfbblack{ 360}{ 0}{25}{0.25}
\lfb{ 360}{ 25}{23}{0.23}
\lfb{ 360}{ 48}{20}{0.20}
\lfb{ 360}{ 68}{15}{0.15}
\lfb{ 360}{ 83}{10}{0.10}
\lfb{ 360}{ 93}{9}{0.09}
  
  \draw [decorate,decoration={brace,amplitude=5pt,mirror},xshift=0pt,yshift=0pt] (0,-5) -- (350,-5) node [midway,yshift=-.5cm] {\small LS-bins};
  \draw [decorate,decoration={brace,amplitude=5pt,mirror},xshift=0pt,yshift=0pt] (360,-5) -- (385,-5) node [midway,yshift=-.5cm] {\small S-bin};
\end{tikzpicture}
}

\bigskip

\ \ \ \ \ \raisebox{1.58ex}{
\scalebox{0.76}{
\begin{tikzpicture}[scale=0.04]
  \lfbblack{  0}{ 0}{24}{0.24}
  \dgbgood{0}{24}{90}{0.90}
  \lfbblack{30}{0}{25}{0.25}	    
  \dgbgood{30}{25}{78}{0.78}
  \lfb{ 60}{0}{10}{0.10}
  \lfb{ 60}{10}{14}{0.14}
  \dgb{ 60}{24}{75}{0.75}	
    \dgb{ 60}{99}{68}{0.68}	

  \lfb{ 90}{0}{16}{0.16}
  \lfb{ 90}{16}{9}{0.09}
  \dgb{90}{25}{73}{0.73}
  \dgb{90}{98}{71}{0.71}

  \dgbgood{120}{0}{85}{0.85}
  \dgb{120}{85}{74}{0.74}	
  
  \dgbgood{150}{0}{77}{0.77}	
  \dgb{150}{77}{65}{0.65}

  \dgb{180}{0}{64}{0.64}	
  \dgb{180}{64}{63}{0.63}	
  
  \lfb{210}{0}{12}{0.12}
  \lfb{210}{12}{11}{0.11}
  \lfb{210}{23}{16}{0.16}
  \lfb{210}{39}{8}{0.08}
  \lfbblack{210}{47}{26}{0.26}
  
  \lfb{210}{73}{20}{0.20}
  \lfb{210}{93}{8}{0.08}
  
  \lfb{240}{0}{23}{0.23}
  \lfb{240}{23}{7}{0.07}
  \lfb{240}{30}{15}{0.15}
  \lfb{240}{45}{10}{0.10}
  \lfb{240}{55}{7}{0.07}
  \lfb{240}{62}{12}{0.12}
  \lfb{240}{74}{22}{0.22}
  \lfb{240}{96}{23}{0.23}
  
   \lfb{270}{0}{7}{0.07}
   \lfb{270}{7}{7}{0.07}
   \lfb{270}{14}{15}{0.15}
   \lfb{270}{29}{23}{0.23}
   \lfb{270}{52}{20}{0.20}
   \lfb{270}{72}{15}{0.15}
  \lfb{270}{87}{10}{0.10} 
  \lfb{270}{97}{9}{0.09} 
   
  \draw [decorate,decoration={brace,amplitude=5pt,mirror},xshift=0pt,yshift=0pt] (0,-5) -- (55,-5) node [midway,yshift=-.5cm] {\small black bins};

  \draw [decorate,decoration={brace,amplitude=5pt,mirror},xshift=0pt,yshift=0pt] (60,-5) -- (115,-5) node [midway,yshift=-.5cm] {\small white bins};
  
    \draw [decorate,decoration={brace,amplitude=5pt,mirror},xshift=0pt,yshift=0pt] (0,-25) -- (115,-25) node [midway,yshift=-.5cm] {\small reserved bins};

  \draw [decorate,decoration={brace,amplitude=5pt,mirror},xshift=0pt,yshift=0pt] (210,-5) -- (295,-5) node [midway,yshift=-.5cm] {\small S-bins};

\end{tikzpicture}
}
}
\caption{(top) An optimal packing with 12 LS-bins ($\mls=12$ and hence $\ngood=12/3=4$). The four good items are colored dark gray. The smallest good item has size $0.77$ and hence $d=1-0.77=0.23$. The approximate value of $d$ is encoded in 6 bits, i.e., $d^+ = 1/8+1/16+1/32+1/64 = 0.234375$. Black items are those no smaller than $0.234375$ and they are colored dark green (they are 0.24, 0.25, and 0.26). There are $\mlsb = 2$ black items in LS-bins. (bottom) A packing that covers more bins than half of the optimal packing. There are $\ngood=4$ reserved bins. We have $\nb= \min\{\mlsb=2, \ngood=4\} =2$. So, the two smallest black items (0.24 and 0.25) are placed in the reserved bins (black bins). The other 2 reserved bins (white bins) are covered with small items up to an approximate level $d$ using the Dual-Worst-Fit algorithm. The remaining small items are placed using Dual-Next-Fit in separate S-bins. Assuming that $\alpha=1/6$, there are $\alpha \mls = 2$ reserved bins that receive good items. The remaining large items are all paired in their bins.}
\label{fig:optvsdes2}
\end{center}
\end{figure}
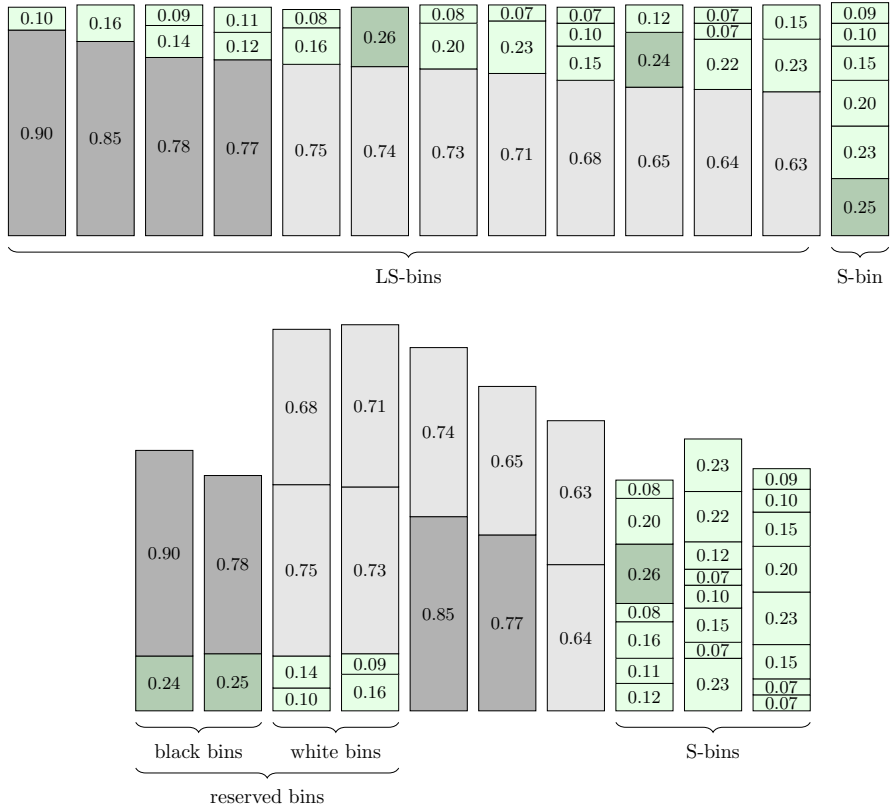

The above algorithm is strictly better than 0.5-competitive. Here is an informal argument for this claim.
The total size of small items placed in S-bins cannot be much smaller than the total size of small items placed in S-bins in an optimal packing. This is due to the following:
  \begin{itemize}
  \item The smallest black items are placed in the reserved bins; the total size of these items is 
  no more than the total size of black items placed with a good item in the LS-bins of \opt.
   \item Each white bin receives white items of total size approximately $\leveld$ (assuming $\ngood>\nb$, i.e., a white bin exists). There are enough white items in the input sequence to ensure this, because 
     \opt covers $\mls - \mlsb$ LS-bins with white items together with a large item. By the assumption that $\ngood>\nb$, $\mlsb = \nb$. Thus, \opt has $\mls - \nb$ LS-bins without a black item. In at most $\mls/3$ of these bins, the total size of white items is less than $\leveld$ (these are bins with a good item). The total size of white items in any of the remaining $2\mls/3-\nb$ bins is at least $\leveld$, that is, the total size of white items in all bins is at least $(2\mls/3-\nb) \times \leveld \geq (\mls/3-\nb) \times 2\leveld$. In the above scheme, white items are placed in $\ngood - \nb \approx \mls/3 - \nb$ bins, that is, on average, a size of at least $2\leveld$ is available for each white bin. Since Dual-Worst-Fit tends to balance the load of white items among the bins, the level of each bin will not be less than $\leveld$, which is half of the average available size (it will be formally proved later).
\end{itemize}
For bins where we use Dual-Next-Fit, each bin has a level less
than~3/2.  Thus, the number of the non-reserved bins covered by small
items is roughly lower-bounded by $2\ms/3$.  Reserved bins that
received a good item are all covered; so, at least $\alpha \mls$ bins
are covered by a large and some small items. Finally, other large
items are placed pairwise in new bins.  We conclude that the resulting
packing is $(\alpha,\epsilon)$-desirable and hence (for $\alpha > 0$
and $\eps < \frac14$) strictly better than 0.5-competitive by
Lemma~\ref{lem:alphades}.

Unfortunately however, using only $O(\log\log n)$ bits of advice, the oracle cannot
give the exact values of \leveld and \mlsb nor identify the \mlsb smallest
black items.
Thus, we make a number of adjustments to the above algorithm as
described in the two following subsections, resulting in a value of
$\alpha$ smaller than $1/15$.

\subsubsection{Placing the small items}

Let \sblack denote the size of the \nb'th smallest black item in the
input sequence.
Let \eblack denote the number of items among the \nb smallest black
items which have sizes in the range $(\sblackm,\sblackp]$. 	
Let \nblack denote the number of black items no larger than \sblackm
and let $\mblack = \nblack + \eblackm$.
The numbers \sblackm, \eblackm, and \mblackm are given as advice, and
the algorithm opens \mblackm reserved black bins. This ensures that the number of
black items used in the algorithm's reserved bins is close to \nb, but not
larger than \nb. 
Each black item of size at most \sblackm is packed in a black bin.  These black items are among the \nb smallest black items. 
Each black item with a size in $(\sblackm,\sblackp]$ is also packed in a
black bin, if less than \eblackm black items in the size range
$(\sblackm,\sblackp]$ have already been placed in black bins.
These black items might not be among the \nb smallest black items but are at most a factor $(1+\epsilon)$ larger. So, $\mblackm$ reserved bins are covered by black items that are not much larger than the smallest $n_B$ black items. 

Let $\ml = \ngood - \nb$.
The number \mlm is given as advice, and the algorithm opens \mlm 
reserved white bins.
As long as the total size of white items received so far is at most
$2\leveldp \mlm$, white items are packed in the white
bins using Dual-Worst-Fit.
The remaining small items are packed in S-bins using Dual-Next-Fit.

We let $\mres = \mblackm+\mlm$ denote the total number of reserved bins.
Note that
\begin{align}
\label{eq:mres}
\mres
& = \mblackm+\mlm \nonumber \\
& = \rounddown{(\nblack + \eblackm)} + \rounddown{(\ngood - \nb)} \nonumber \\
& > (1-\eps)(\nblack + \eblackm) + (1-\eps)(\ngood - \nb) \nonumber \\
& > (1-\eps)^2 \nb + (1-\eps)(\ngood - \nb) \nonumber \\
& > (1-\eps)^2 \ngood
\end{align}

\begin{lemma}\label{lem:goodreseved}
  A reserved bin that includes any good item will be covered in the
  final packing of the algorithm.
\end{lemma}

\begin{proof}
Since any good item has size at least $1-\leveld$, we just need to
prove that each reserved bin receives small items of total size at
least \leveld.
Since $\leveldp \geq \leveld$, the claim is true for all black bins.
Thus, we only need to consider the white bins.
For each white bin, we let the {\em white level} of the bin denote the
total size of white items packed in the bin.

Recall that Dual-Worst-Fit places each item in a bin with minimum white
level.
This means that, if a white item is ever placed in a bin that already
has a white level of at least \leveld, all white bins have a white
level of at least \leveld.
Assume for the sake of contradiction that no white item is
placed in a bin that already has a white level of at least \leveld.

The algorithm stops packing white items in the
white bins when 
\begin{itemize}
\item[(a)] we have reached the upper bound of $2\leveldp \mlm$ total
  size, or
\item[(b)] there are no more white items in the sequence,
\end{itemize}
whichever happens first.
In both cases, we arrive at a contradiction.

We first consider (a), i.e., we assume that white items of total size
at least $2\leveldp \mlm$ are packed in the white bins.
By the assumption that no white item is placed in a bin that already
has a white level of at least \leveld, the final white level of each
of the \mlm white bins is less than $\leveld + \leveldp$, a
contradiction.

Now, consider (b), i.e., assume that all white items of the input
sequence are packed in white bins.
Assume that exactly $\ell$ white items packed in white bins are larger than
\leveld, and let $\leveld+e_1,\ldots,\leveld+e_{\ell}$ be the sizes of
these items.
By the assumption that no white item is placed in a bin with a white
level of at least~$\leveld$, 
these $\ell$ items are packed in bins
with a final white level less than $2\leveld+e_i$, $1 \leq i \leq
\ell$, respectively.
By the same assumption, the remaining $\mlm-{\ell}$ white bins each
have a white level less than $2\leveld$.
Thus, the total size of white items is
\begin{align}
\label{eq:whitesize}
  \sizel < (\mlm-\ell) \cdot 2\leveld + \sum_{i=1}^{\ell} (2d+e_i)
             = \mlm \cdot 2 \leveld + \sum_{i=1}^{\ell} e_i
\end{align}

\opt covers \mls bins, each with a large item and some small items.
Out of the \mls large items used for these bins, fewer than \ngood
items are larger than
$1-d$ (by the definition of \leveld). 
Thus, at least $\mls-\ngood \geq 2\ngood$ of these large
items have size at most $1-d$, and at most $\nb$ of them are combined
with a black item in the optimal packing, as $\nb \geq \mlsb$.
Therefore \opt has at least $2\ngood - \nb \geq 2\mlm$ bins, each with
a total size of white items of at least~$\leveld$.
Hence, the total size of white items is
\begin{align*}
\sizel & \geq (2\mlm - \ell) \cdot \leveld + \sum_{i=1}^{\ell} (d+e_i)
         =  2\mlm \cdot \leveld + \sum_{i=1}^{\ell} e_i\,,
\end{align*}
contradicting Equation~(\ref{eq:whitesize}).
\end{proof}

We now prove that the algorithm fulfills Property~III of an
$(\alpha,\eps)$-desirable packing:

\begin{lemma}
\label{lemma:sbins}
The algorithm covers at least $\frac23\ms- O(\eps\mls)$ S-bins.
\end{lemma}

\begin{proof}
We first show that the total size of small items packed in S-bins in
the algorithm's packing is at least $\ms- O(\eps\mls)$.  The total size
of items that \opt packs in S-bins is at least \ms.  Thus, since all
small items that are not packed in reserved bins are packed in S-bins,
it is sufficient to argue that the total size of small items packed in
the reserved bins is at most a factor of $O(1+\eps)$ larger than the
total size of the small items packed in the LS-bins of \opt.

As argued in the proof of Lemma~\ref{lem:goodreseved}, \opt packs
white items of total size at least $2\leveld\mlm$ in its LS-bins.
Since our algorithm stops packing white items after total size at least
$2\leveldp\mlm$ has been reached, it packs white items of less than $2\leveldp\mlm +\leveldp$ total
size in its reserved bins. Thus, the algorithm uses less than
\begin{align*}
2\leveldp\mlm+\leveldp - 2\leveld\mlm
  & < 2(1+\eps)\leveld\mlm - 2\leveld\mlm + \leveldp \\
  & = 2\eps\leveld\mlm +\leveldp \\
  & < 2\eps \mlm +\leveldp \\
  & \leq 2\eps \mls +\leveldp
\end{align*}
more total size of white items for its reserved bins as does \opt for
its LS-bins.

The algorithm packs no more black items in its reserved bins than does
\opt in its LS-bins.
Furthermore, these items are not much larger than those \opt uses in LS-bins
since at most \eblackm  black items in the size interval $(\sblackm,\sblackp]$
are used in those reserved bins.
Thus, since the black items no larger than \sblackm are the smallest black
items in the input, the total size of black items packed in the
reserved bins is at most $1+\eps$ times the total size of items
packed in \opt's LS-bins.

This proves that the algorithm packs small items of total size at
least $\ms- O(\eps\mls)$ in its S-bins.
Since no items larger than $1/2$ are packed in the S-bins, no bin gets
filled to more than $3/2$.
This means that at least $2\ms/3 - O(\eps\mls)$ bins are covered.
\end{proof}

\subsubsection{Placing the large items}
In order to achieve an $(\alpha,\eps)$-desirable packing, our
algorithm needs to select approximately $\alpha \mls \leq \mres$ good
items, and place them in the reserved bins;
Lemma~\ref{lem:goodreseved} guarantees that these bins will be covered
(Property~I holds).  Note that Lemma~\ref{lemma:sbins} establishes
Property~III.  Meanwhile, as we will describe, the algorithm ensures
that other large items are paired and hence each pair of them covers a
bin (Property~II holds).

In order to pack at least $\alpha \mls - O(\eps\mls)$ good items in the
reserved bins, the algorithm considers three cases depending on the
location of good items in the sequence. Advice will be used to select
the correct case.
Furthermore, \mlsm will be given as advice.

\begin{lemma}
\label{smallbeta}
For $\beta<15/14$, there exists an online algorithm using $O(\log\log
n)$ advice bits and producing an $(\alpha, \eps)$-desirable packing
for
$\alpha \leq \frac{(1-\eps)^2(\frac76-\beta)}{1 +
  \frac32(1-\eps)^2(1+\eps)}$ and $\eps \in O(1/(\log n))$.
\end{lemma}
\begin{proof}
  By Lemma~\ref{lemma:sbins}, Property~III holds.

  To prove that Property~I holds, we first argue that the value of
  $\alpha\mls$ is not larger than the number of reserved bins.  By
  Equation~(\ref{eq:mres}), the number of reserved bins is
  $$\mres > (1-\eps)^2\floor{ \frac{\mls}{3} } > \frac34 \floor{
  \frac{\mls}{3} }, \text{ for } \eps \leq \frac18.$$
  Since $\beta \geq 1$, we get
  $$\alpha \leq \frac{(1-\eps)^2(\frac76-\beta)}{1 +
  \frac32(1-\eps)^2(1+\eps)} < \frac{7-6\beta}{6/(1-\eps)^2 + 9(1+\eps)} < \frac{7-6\beta}{15} < \frac{1}{15}.$$

  Large items are placed in a way to establish Properties~I and~II,
  following a case analysis based on where good items appear in the
  request sequence (the appropriate case is encoded using advice).
  
  \textbf{Case 1:} Assume there are at least $\floor{\alpha \mlsm}$ good items
  among the first $\mres$ large items in the sequence.

  In this case, the algorithm places the first \mres large items into
  the reserved bins.
  After seeing all these \mres items, the algorithm chooses the
  largest $\lfloor \alpha \mlsm \rfloor$ of them and declares them to
  be good items, which by Lemma~\ref{lem:goodreseved} are guaranteed
  to be covered.
  Thus, Property~I holds.
  
  The remaining $\mres - \lfloor \alpha \mlsm \rfloor$ large items in the
  reserved bins will be paired with forthcoming large items. Since
  there are at least $\mls-\mres \geq 2\mls/3$ forthcoming large items and fewer
  than $\mls/3$ large items in the reserved bins waiting to be paired,
  all these large items (except possibly one) can be paired
  (Property~II holds).  In summary, in the final covering, there are
  $\lfloor \alpha \mlsm
  \rfloor$ bins covered by a large item (and some small items) while
  the remaining large items are paired (except possibly one).  Hence,
  Property~II holds.

  \textbf{Case 2:} Assume there are fewer than
  $\floor{\alpha \mlsm}$ good items among the
  first \mres large items in the sequence (Case 1 does not
  apply). Furthermore, assume there are
  at least $\floor{\alpha \mlsm}$ good items among the
  \mres large items that follow the first
  \mres large items.

  In this case, the algorithm places the first \mres large items
  pairwise in $\lceil \mres/2 \rceil$ bins that are not reserved. The
  \mres large items that follow are placed in the reserved bins. After
  placing the last of these items in the reserved bins, the algorithm
  considers these \mres items and declares the $\lfloor \alpha \mlsm
  \rfloor$ largest to be good items, which by
  Lemma~\ref{lem:goodreseved} are guaranteed to be covered.  Thus,
  Property~I holds.

  The
  remaining $\mres- \floor{\alpha \mlsm}$ reserved bins (with
  large items) will need to be covered by forthcoming large items. We
  know there are at least $\mls-2\mres \geq \lfloor \mls/3 \rfloor$ forthcoming
  large items and fewer than $\lfloor \mls/3 \rfloor$ large items in
  reserved bins waiting to be paired, so all these large items (except
  possibly one) can be paired.  Thus, Property~II holds.

  \textbf{Case 3:} Assume there are fewer than
  $\floor{\alpha \mlsm}$ good items among the
  first \mres large items and also fewer than
  $\floor{\alpha \mlsm}$ good items among the
  following \mres items (Cases~1 and~2 do not apply). 

  In this case, the algorithm places the first $2\mres$ large items in
  pairs.  There are
  $L = 2\mll+\mls-2\mres$ remaining
  large items.  The algorithm places the first
  $F=\roundup{\mll}+ \lfloor \mres/2 \rfloor + \lfloor \alpha
    \mlsm/2 \rfloor - 1 $ of the last $L$ large items in the
  reserved bins (note that this is roughly half of the last $L$ large
  items when $\alpha$ is small).  For this to be possible, we show
  first that the number of reserved bins is at least $F$, i.e.,
  $F-\mres \leq 0$.
\begin{align*}
  F - \mres
  &   =  \roundup{\mll}+ \left\lfloor \frac{\mres}{2} \right\rfloor + \left\lfloor \frac{\alpha \mlsm}{2} \right\rfloor - 1 - \mres \\
  & \leq \roundup{\mll} - \frac{\mres}{2} + \frac{\alpha
    \mlsm}{2} - 1 \\ 
  &   <  \roundup{\mll} - (1-\eps)^2\frac{\ngood}{2} + \frac{\alpha
    \mlsm}{2} - 1, \text{ by Equation~(\ref{eq:mres})} \\ 
  &   <  \roundup{\mll} - (1-\eps)^2 \frac{\mls}{6}
         + \frac{\alpha \mlsm}{2} \\ 
  &   < (1+\eps)\mll - (1-\eps)^2\frac{\mls}{6} + \frac{\alpha
    \mls}{2} \\ 
  &   = (1+\eps)(\beta-1)\mls - (1-\eps)^2\frac{\mls}{6} + \frac{\alpha
    \mls}{2} \\ 
  &   =  \left((1+\eps)(\beta-1) - \frac{(1-\eps)^2}{6} + \frac{\alpha}{2} \right) \mls \\
  &   <  \left(\frac{1+\eps}{14} - \frac{(1-\eps)^2}{6}+\frac{1}{30}
         \right) \mls
         \\
         &   < 0, \text{ for } \eps \leq \frac18 ,
\end{align*}
The second to last equality follows since $\mll = (\beta-1)\mls$, and
the second to last inequality follows from $\beta < 15/14$ and $\alpha
< 1/15$.

This establishes that there are enough reserved bins.

Next, we show that at least $\alpha \mls - 6$ of the $F$ items placed in
the reserved bins are good items.  There are fewer than
$2 \floor{\alpha \mlsm}$ good items among the first
$2\mres$ large items. Thus, the number of good items among the last $L$
large items is more than
$\ngood - 2\lfloor\alpha \mlsm\rfloor$.  Even
if all items among the last $L-F$ large items are good, there are
still more than
$ G = \ngood - 2 \floor{\alpha \mlsm} - (L-F)$ good items that are placed in the reserved bins.  In the
following we deduce that $G > \alpha \mls -6$:
\begin{align*}
  G
  & = \ngood - 2 \floor{\alpha \mlsm} - L + F\\
  & = \ngood - 2 \left\lfloor \alpha \mlsm\right\rfloor - \left( 2\mll+\mls-2\mres \right) +  \roundup{\mll}+\left\lfloor \frac{\mres}{2}\right\rfloor +\left\lfloor  \frac{\alpha\mlsm}{2} \right\rfloor -1 \\
  & > \ngood - 2 \alpha \mlsm - \left( 2\mll+\mls- 2\mres \right) + \roundup{\mll}+ \frac{\mres}{2} +  \frac{\alpha\mlsm}{2}-3 \\
  & = \ngood - \frac32 \alpha \mlsm - 2\mll-\mls+ \frac52 \mres + \roundup{\mll} -3 \\
  & > \ngood - \frac32 \alpha \mlsm - 2\mll-\mls+ \frac52 (1-\eps)^2 \ngood + \roundup{\mll} -3 \\
  & > \frac72 (1-\eps)^2 \ngood - \frac32 \alpha \mlsm - \mll-\mls -3 \\
  & > \frac76 (1-\eps)^2 \mls - \frac32 \alpha \mls - \mll-\mls -6 \\
  & = \frac76 (1-\eps)^2 \mls - \frac32 \alpha \mls - \beta \mls -6 \\
  & > (1-\eps)^2 \left( \frac 76 - \frac32 \alpha- \beta \right) \mls - 6
\end{align*}
The second to last inequality follows since
$\ngood = \floor{\frac{\mls}{3}} \geq \frac{\mls}{3}-\frac23$.

Since we assumed 
$\alpha \leq \frac{(1-\eps)^2(\frac76-\beta)}{1 + \frac32(1-\eps)^2(1+\eps)}$, we can write
$$\beta \leq \frac76 - \frac{1 + \frac32(1-\eps)^2(1+\eps)}{(1-\eps)^2} \alpha =
  \frac76 - \frac{\alpha}{(1-\eps)^2} -  \frac32(1+\eps)\alpha ,$$
and
\begin{align*}
  G & > (1-\eps)^2 \left(  \frac 76 - \frac32 \alpha - \frac76 + \frac{\alpha}{(1-\eps)^2} + \frac32(1+\eps)\alpha \right) \mls - 6\\
    & > \alpha \mls -6 .
\end{align*}

Hence, at least $\lfloor \alpha \mls \rfloor -6$ of the $F$ items placed
in the reserved bins are good items. After placing these $F$ items,
the algorithm declares the largest $\lfloor \alpha \mlsm \rfloor -6$ among
them to be good items. By Lemma~\ref{lem:goodreseved}, these items
(along with small items in the reserved bins) will cover their
respective bins.
Thus, Property~I holds.

There are $F- \lfloor \alpha \mlsm \rfloor + 6$
large items in
reserved bins which have not been declared good, and the $L - F$ large
items which have not arrived at this point will be paired with them.
We prove that, among the $F- \lfloor \alpha \mlsm \rfloor + 6$
large items that are not declared good, only $O(\eps\mls)$ items might not be paired with
one of the remaining $L-F$ large items:
\begin{align*}
& F - \floor{\alpha \mlsm} + 6 \\
& = \mllp+ \left\lfloor\frac{\mres}{2}\right\rfloor+ \floor{ \frac{\alpha \mlsm}{2} } - 1 - \floor{ \alpha \mlsm } + 6 \\
& \leq \mllp+ \left\lfloor\frac{\mres}{2}\right\rfloor - \floor{ \frac{\alpha \mlsm}{2} } + 6\\
& \leq (2\mllp - \mllp) + \left(3\mres-2\mres-\left\lfloor\frac{\mres}{2}\right\rfloor\right)- \floor{ \frac{\alpha \mlsm}{2} } + 6\\
& \leq (2\mll - \mllp) + \left(\mls-2\mres-\left\lfloor\frac{\mres}{2}\right\rfloor\right)- \floor{ \frac{\alpha \mlsm}{2} } + O(\eps\mll)\\
& < \left(2\mll + \mls - 2 \mres\right) - \left(\mllp+\left\lfloor\frac{\mres}{2}\right\rfloor+ \floor{ \frac{\alpha \mlsm}{2}} - 1\right) + O(\eps\mll)\\
& = L-F + O(\eps\mll)\\
& = L-F + O(\eps\mls), \text{ since } \mll=(\beta-1)\mls < \frac{1}{14}\mls
\end{align*}

So, the number of large items in the reserved bins which are not
paired will be $O(\eps\mls)$. The bins in which these unpaired items
are placed, along with the $\lfloor \alpha \mls \rfloor -6$ bins that
include good items, will be the LS-bins in the final
$(\alpha,\eps)$-desirable packing.  All large items placed in bins
other than LS-bins are paired and hence, Property~II also holds.
\end{proof}

\subsection{Wrapping it up}

\begin{theorem}\label{lem:hardcaselog}
  There is an algorithm that, provided with $O(\log \log n)$ bits of
  advice, achieves an asymptotic competitive ratio of at least
  $\frac{8}{15}$.
\end{theorem}
\begin{proof}
  The advice indicates the values of \mlsm, \mllp, \leveldp, \sblackm,
  \sblackp, \eblackm, \mblackm, and \mlm.  These values can all be
  encoded in $O(\log \log n)$ bits of advice.  Note that one cannot
  calculate $\beta$ exactly, since $\mls$ and $\mll$ are not known
  exactly. Thus, the advice also includes $1$ bit to indicate if
  Lemma~\ref{lem:easycaselog} should be used because $\beta$ is larger
  than $15/14$, or $\mls = 0$. If not, the advice also indicates one
  of the three cases described above; this requires two more
  bits. Thus, the advice is $O(\log \log n)$ bits.

  If Lemma~\ref{lem:easycaselog} is used, the competitive ratio is at
  least $8/15$.  Otherwise, provided with this
  advice and a sufficiently small \eps, the algorithm
  can create an $(\alpha,\eps)$-desirable packing of the input sequence for any
  $\alpha \leq \frac{(1-\eps)^2(\frac76-\beta)}{1 +
    \frac32(1-\eps)^2(1+\eps)}$.
  By Lemma~\ref{lem:alphades}, the resulting packing has an asymptotic
  competitive ratio of at least
  $\frac{\alpha+2\beta-1}{2\beta}$, for $\eps \in o(1)$. Choosing
  $\alpha = \frac{(1-\eps)^2(\frac76-\beta)}{1 +
    \frac32(1-\eps)^2(1+\eps)}$ and $\eps \leq 1/\log n$ gives an algorithm with asymptotic competitive ratio at least
  $\frac{12\beta -4}{15\beta} \geq \frac{8}{15}$.
\end{proof}

\section{Impossibility result for advice of sub-linear size}\label{sec:sublin}
This section uses what is normally referred to as lower bound
techniques, but since our ratios are smaller than 1, an upper bound is
a negative result, and we refer to such results as negative or
impossibility results.
In what follows, we show that, in order to achieve any competitive
ratio larger than 15/16, advice of linear size is necessary. We use a
reduction from the binary separation problem, introduced
in~\cite{BoyarKLL16}:

\begin{definition}
  The {\em Binary Separation Problem} is the following online
  problem. The input
  $I = (n_1, \sigma= \left\langle y_1, y_2, \ldots, y_n
  \right\rangle)$ consists of $n = n_1+n_2$ positive values which are
  revealed one by one. There is a fixed partitioning of the set of
  values into a subset of $n_1$ {\em large} values and a subset of $n_2$
  {\em small} values, so that all large values are larger than all small
  values. Upon receiving a value $y_i$, an online algorithm must guess
  if $y$ belongs to the set of small or large values. After the
  algorithm has made a guess, it is revealed to the algorithm which
  class $y_i$ belongs to.
  The number $n_1$ is known to the algorithm from the beginning. 
\end{definition}

A reduction from a closely related problem named ``binary string
guessing with known history'' shows that, in order to guess more than
half of the values correctly, advice of linear size is required:

\begin{lemma}[\cite{BoyarKLL16}]\label{lem:seplow}
  For any fixed $\beta > 0$, any deterministic algorithm for the
  Binary Separation Problem that is guaranteed to guess correctly on
  more than $(1/2+\beta) n$ input values on an input of length~$n$
  needs at least $\Omega(n)$ bits of advice.
\end{lemma}

The following lemma provides the actual reduction from the Binary
Separation Problem to bin covering.

\begin{lemma}\label{lem:mainred}
  Consider the bin covering problem on sequences of length~$2n$ for
  which \opt covers $n$ bins.  Assume that there is an online
  algorithm \A that solves the problem on these instances using $b(n)$
  bits of advice and covers at least $n-r(n)/8$ bins. Then there is
  also an algorithm \BSA that solves the Binary Separation Problem on
  sequences of length~$n$ using $b(n)+O(\log n)$ bits of advice and guessing
  incorrectly at most $r(n)$ times.
\end{lemma}

\begin{proof}
  In the reduction, we encode requests for the algorithm
  \BSA as items for bin covering, which will be given to the algorithm
  \A.  Assume we are given an instance
  $I = (n_1, \sigma= \left\langle y_1, y_2, \ldots, y_n\right\rangle)$
  of the Binary Separation Problem, in which $n_1$ is the number of
  large values ($n_1+n_2 = n$), and the values $y_t$ are revealed
  in an online manner $(1 \leq t \leq n)$. We create an instance of
  the bin covering problem for \A which has length~$2n$.

  The bin covering sequence starts with $n_1$ ``huge'' items of size
  $1-\epsilon$ for some $\epsilon < \frac{1}{2n}$ (this will ensure
  that at least one item of size at least $1/2$ is in every covered
  bin). The value $n$ is given as advice, using $O(\log n)$ bits.
  In the optimal covering, each of these huge items will be
  placed in a different bin.  The next $n$ items are created in an
  online manner, so that we can use the result of their packing to
  guess the requests for the Binary Separation Problem. Let
  $\tau = y_t$ $(1 \leq t \leq n)$ be a requested value from the Binary
  Separation Problem, and choose an increasing function
  $f : \REALS^+ \rightarrow (\epsilon,2\epsilon)$.  When $\tau$ is
  presented to \BSA, $f(\tau)$ is presented to \A, and the item of size
  $f(\tau)$ is said to be associated with the request $\tau$.  If the
  algorithm places the item $f(\tau)$ in one of the bins opened by the
  huge items, \BSA will answer that $\tau$ is small; otherwise, \BSA
  will answer that $\tau$ is large.
  
  The last $n_2$ items of the bin covering instance are defined as
  complements of the items in the bin covering instance associated
  with large values in the binary separation instance (the complement
  of item $x$ is $1-x$). We do not need to give the last items
  complementing the small items in order to implement the algorithm
  \BSA, but we need them for the proof of the quality of the
  correspondence that we are proving.

  Call an item in the bin covering sequence ``large'' if it is
  associated with a large value in the Binary Separation Problem, and
  ``small'' if associated with a small value.  For the bin covering
  sequence produced by the reduction, an optimal algorithm pairs each
  of the small items with a huge item,
  placing it in one of the first $n_1$ bins. \opt pairs the large
  items with their complements, starting one of the next $n_2$ bins
  with each of these large items. Hence, the number of bins in an
  optimal covering is $n_1+n_2=n$.

  Let $a_1$ and $a_2$ denote the number of the two types of mistakes
  that the bin covering algorithm \A makes, causing incorrect answers
  to be given by \BSA.  Thus, we let $a_1$ denote the number of small
  items which do not get placed with a huge item, and $a_2$ denote the
  number of large items which are placed with a huge item.  Clearly,
  the number of errors in the answers provided by the binary
  separation algorithm \BSA is $a_1+a_2$. We claim that the number of
  bins covered by the bin covering algorithm \A is at most
  $n_1+n_2 -(a_1+a_2)/8$. \A covers at most $(n_1-a_1)+(n_2-a_2)$ bins
  in the same way that \opt does; these bins include either a huge
  item with a small item or a large item and its complement. Moreover,
  at most $\min\{a_1,a_2\}$ bins are covered by huge items that are
  placed with large items, but no small item. Similarly, at most
  $\min\{a_1/2,a_2\}$ bins are covered with two small items and a
  complement of a large item. Finally, there are $p=\max\{a_1-a_2,0\}$
  bins with huge items that are not covered with any small or large
  items; similarly, there are $q=\max\{a_2 -a_1/2,0\}$ complements of
  large items that are not covered by large or pairs-of-small
  items. These $p+q$ items can be paired to cover at most $(p+q)/2$
  bins.  In total, the number of bins that are covered by \A is
\[\begin{array}{cl}
  \alg(I) \leq  & (n_1-a_1) + (n_2-a_2) + \min\{a_1,a_2\} + \min\{a_1/2,a_2\} \\[1ex]
     & + \frac12(\max\{a_1-a_2,0\} + \max\{a_2-a_1/2,0\})
\end{array}\]

If $a_1\leq a_2$, the above value becomes $n_1+n_2-a_2/2 +a_1/4$ which
is indeed at most $n_1+n_2-(a_1+a_2)/8$. If $a_1/2\leq a_2 \leq a_1$,
the above value becomes $n_1+n_2-a_1/4$ which is at most
$n_1+n_2-(a_1+a_2)/8$. Finally, if $a_2< a_1/2$, the above value
becomes $n_1+n_2-(a_1+a_2)/2$ which is less than
$n_1+n_2-(a_1+a_2)/8$. In summary, when the algorithm \BSA makes
$a_1+a_2$ errors in partitioning small and large items, \A covers at
most $n-(a_1+a_2)/8$ bins (intuitively speaking, each eight mistakes
in binary guessing causes at least one fewer bin to be covered; see
Figure~\ref{fig:twopackings}). In other words, if the number of covered
bins is at least $n - r(n)/8$, then the number of binary separation
errors must be at most $r(n)$.
\end{proof}

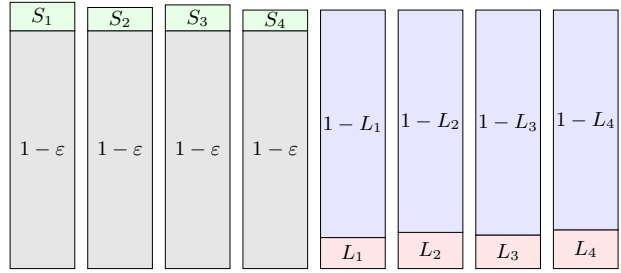
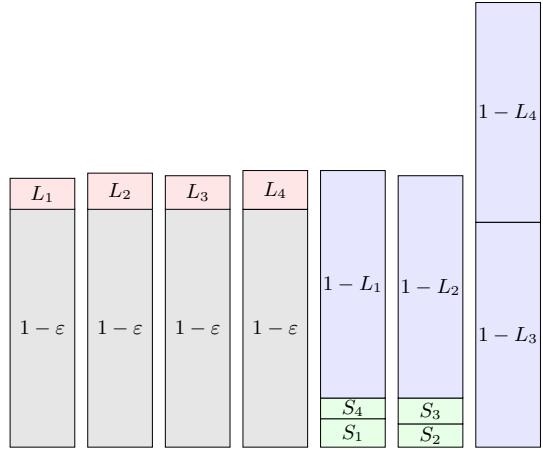
\begin{figure*}[t!]
    \centering
    \begin{subfigure}[t]{0.7\textwidth}
\scalebox{0.86}{
\begin{tikzpicture}[scale=0.04]
  \dgb{  0}{ 0}{92}{1-\epsilon}
  \lfb{  0}{92}{11}{S_1}
  \dgb{ 30}{ 0}{92}{1-\epsilon}
  \lfb{ 30}{92}{ 9}{S_2}
  \dgb{ 60}{ 0}{92}{1-\epsilon}
  \lfb{ 60}{92}{10}{S_3}
  \dgb{ 90}{ 0}{92}{1-\epsilon}
  \lfb{ 90}{92}{ 8}{S_4}
  \rb{120}{ 0}{12}{L_1}
  \bb{120}{12}{88}{1-L_1}
  \rb{150}{ 0}{14}{L_2}
  \bb{150}{14}{86}{1-L_2}
  \rb{180}{ 0}{13}{L_3}
  \bb{180}{13}{87}{1-L_3}
  \rb{210}{ 0}{15}{L_4}
  \bb{210}{15}{85}{1-L_4}
\end{tikzpicture}
}
\caption{The optimal covering}\label{subfig:opt}
\end{subfigure}

\bigskip
    
\begin{subfigure}[t]{0.7\textwidth}
\scalebox{0.86}{
\begin{tikzpicture}[scale=0.04]
  \dgb{  0}{ 0}{92}{1-\epsilon}
  \lfb{120}{ 0}{11}{S_1}
  \dgb{ 30}{ 0}{92}{1-\epsilon}
  \lfb{150}{ 0}{ 9}{S_2}
  \dgb{ 60}{ 0}{92}{1-\epsilon}
  \lfb{150}{ 9}{10}{S_3}
  \dgb{ 90}{ 0}{92}{1-\epsilon}
  \lfb{120}{11}{ 8}{S_4}
  \rb{  0}{92}{12}{L_1}
  \bb{120}{19}{88}{1-L_1}
  \rb{ 30}{92}{14}{L_2}
  \bb{150}{19}{86}{1-L_2}
  \rb{ 60}{92}{13}{L_3}
  \bb{180}{0}{87}{1-L_3}
  \rb{ 90}{92}{15}{L_4}
  \bb{180}{87}{85}{1-L_4}
\end{tikzpicture}
}
\caption{A covering with 8 mistakes}\label{subfig:online}
\end{subfigure}
\caption{Two coverings for the sequences used in reduction from binary
  separation to bin covering. Grey items show huge items that arrive
  first. Red and green items are respectively large and small items
  that need to be separated. The sequence ends with complements of
  large items. (a) shows an optimal packing and (b) shows a covering
  with eight mistakes; as a result of the eight mistakes, one fewer
  bin is covered.}
\label{fig:twopackings}
\end{figure*}

It turns out that reducing the Binary Separation Problem to bin
covering (the above lemma) is more involved than a similar reduction
to the bin packing problem~\cite{BoyarKLL16}.  The difference roots in
the fact that there are more ways to place items into bins in the bin
covering problem compared to bin packing; this is because many
arrangements of items are not allowed in bin packing due to the
capacity constraint.

\begin{theorem}\label{th:linlow}
  Consider the bin covering problem on sequences of length $n$. To
  achieve a competitive ratio of $15/16+\delta$, in which $\delta$ is
  a small, but fixed positive constant, an online algorithm needs to
  receive $\Omega(n)$ bits of advice.
\end{theorem}

\begin{proof}
Consider a bin covering algorithm \A with competitive ratio
$15/16+\delta$, and sequences of length $2n$ for which $\opt$ covers
$n$ bins. \A covers $(15/16+\delta)n = n -r(n)/8$ bins for $r(n) =
(1/2 - 8\delta)n$. Applying Lemma~\ref{lem:mainred}, we conclude
that there is an algorithm that solves the Binary Separation Problem
on sequences of length $n$ using at most $O(\log n)$ bits of
additional advice, while making at most $(1/2-8\delta)n$ errors. By
Lemma~\ref{lem:seplow}, we know that such an algorithm requires
$\Omega(n)$ bits of advice. Thus, \A requires $\Omega(n)$ bits of
advice as well.
\end{proof}

\section{Concluding remarks}
We have established that for bin covering $\Theta(\log\log n)$ bits of advice are
necessary and sufficient to improve the competitive ratio obtainable
by purely online algorithms. This differs significantly from the
results from bin packing, where a constant number of bits of advice are sufficient.

Obvious questions are: How much better than our bound of
$8/15=0.53\bar{3}$ can one do with $O(\log\log n)$ bits of advice?
Can one do better with $O(\log n)$ bits of advice?

\section{Acknowledgments}
We thank anonymous referees for their valuable comments.

\bibliographystyle{plain}
\bibliography{refs}

\begin{thebibliography}{10}

\bibitem{ADKRR18}
Spyros Angelopoulos, Christoph D{\"{u}}rr, Shahin Kamali, Marc~P. Renault, and
  Adi Ros{\'{e}}n.
\newblock Online bin packing with advice of small size.
\newblock {\em Theory of Computing Systems}, 62(8):2006--2034, 2018.

\bibitem{AsgeirssonACEMPVWW02}
Eyj{\'{o}}lfur~Ingi {\'{A}}sgeirsson, Urtzi Ayesta, Edward G.~Coffman Jr.,
  J.~Etra, Petar Momcilovic, David~J. Phillips, V.~Vokhshoori, Z.~Wang, and
  J.~Wolfe.
\newblock Closed on-line bin packing.
\newblock {\em Acta Cybernetica}, 15(3):361--367, 2002.

\bibitem{AssmannJKL84}
Susan~F. Assmann, David~S. Johnson, Daniel~J. Kleitman, and Joseph~Y.{-}T.
  Leung.
\newblock On a dual version of the one-dimensional bin packing problem.
\newblock {\em Journal of Algorithms}, 5(4):502--525, 1984.

\bibitem{balogh_et_al}
J{\'a}nos Balogh, J{\'o}zsef B{\'e}k{\'e}si, Gy{\"o}rgy D{\'o}sa, Leah Epstein,
  and Asaf Levin.
\newblock A new and improved algorithm for online bin packing.
\newblock In {\em 26th Annual European Symposium on Algorithms (ESA)}, volume
  112 of {\em Leibniz International Proceedings in Informatics (LIPIcs)}, pages
  5:1--5:14. Schloss Dagstuhl--Leibniz-Zentrum fuer Informatik, 2018.

\bibitem{Balogh_et_al_lowerBound18}
J{\'{a}}nos Balogh, J{\'{o}}zsef B{\'{e}}k{\'{e}}si, Gy{\"{o}}rgy D{\'{o}}sa,
  Leah Epstein, and Asaf Levin.
\newblock A new lower bound for classic online bin packing.
\newblock {\em ArXiv}, 1807.05554 [cs:DS], 2018.

\bibitem{BockenhauerHK14}
Hans{-}Joachim B{\"{o}}ckenhauer, Juraj Hromkovic, and Dennis Komm.
\newblock A technique to obtain hardness results for randomized online
  algorithms -- {A} survey.
\newblock In {\em Computing with New Resources -- Essays Dedicated to Jozef
  Gruska on the Occasion of His 80th Birthday}, volume 8808 of {\em Lecture
  Notes in Computer Science}, pages 264--276. Springer, 2014.

\bibitem{BockenhauerHKKSS14}
Hans{-}Joachim B{\"{o}}ckenhauer, Juraj Hromkovi{\v{c}}, Dennis Komm, Sacha
  Krug, Jasmin Smula, and Andreas Sprock.
\newblock The string guessing problem as a method to prove lower bounds on the
  advice complexity.
\newblock {\em Theoretical Compututer Science}, 554:95--108, 2014.

\bibitem{BKKK17}
Hans{-}Joachim B{\"{o}}ckenhauer, Dennis Komm, Rastislav Kr{\'{a}}lovic, and
  Richard Kr{\'{a}}lovic.
\newblock On the advice complexity of the k-server problem.
\newblock {\em Journal of Computer and System Sciences}, 86:159--170, 2017.

\bibitem{BKKKM17}
Hans{-}Joachim B{\"{o}}ckenhauer, Dennis Komm, Rastislav Kr{\'{a}}lovic,
  Richard Kr{\'{a}}lovic, and Tobias M{\"{o}}mke.
\newblock Online algorithms with advice: The tape model.
\newblock {\em Information and Computation}, 254:59--83, 2017.

\bibitem{BorodinPS18}
Allan Borodin, Denis Pankratov, and Amirali Salehi{-}Abari.
\newblock A simple {PTAS} for the dual bin packing problem and advice
  complexity of its online version.
\newblock In {\em 1st Symposium on Simplicity in Algorithms (SOSA)}, LIPIcs,
  pages 8:1--8:12. Schloss Dagstuhl - Leibniz-Zentrum fuer Informatik, 2018.

\bibitem{BoyarFKLM17}
Joan Boyar, Lene~M. Favrholdt, Christian Kudahl, Kim~S. Larsen, and Jesper~W.
  Mikkelsen.
\newblock Online algorithms with advice: {A} survey.
\newblock {\em {ACM} Computing Surveys}, 50(2):19:1--19:34, 2017.

\bibitem{BoyarKLL16}
Joan Boyar, Shahin Kamali, Kim~S. Larsen, and Alejandro L{\'{o}}pez{-}Ortiz.
\newblock Online bin packing with advice.
\newblock {\em Algorithmica}, 74(1):507--527, 2016.

\bibitem{CsirikT88}
J{\'{a}}nos Csirik and V.~Totik.
\newblock Online algorithms for a dual version of bin packing.
\newblock {\em Discrete Applied Mathematics}, 21(2):163--167, 1988.

\bibitem{DurrKR16}
Christoph D{\"{u}}rr, Christian Konrad, and Marc~P. Renault.
\newblock On the power of advice and randomization for online bipartite
  matching.
\newblock In {\em 24th Annual European Symposium on Algorithms (ESA)},
  volume~57 of {\em LIPIcs}, pages 37:1--37:16. Schloss Dagstuhl -
  Leibniz-Zentrum fuer Informatik, 2016.

\bibitem{Grove95}
Edward~F. Grove.
\newblock Online bin packing with lookahead.
\newblock In {\em 6th Annual {ACM-SIAM} Symposium on Discrete Algorithms
  (SODA)}, pages 430--436. {SIAM}, 1995.

\bibitem{HR17}
Rebecca Hoberg and Thomas Rothvoss.
\newblock A logarithmic additive integrality gap for bin packing.
\newblock In {\em 28th Annual ACM-SIAM Symposium on Discrete Algorithms
  (SODA)}, pages 2616--2625. {SIAM}, 2017.

\bibitem{HKK10}
Juraj Hromkovic, Rastislav Kr{\'{a}}lovic, and Richard Kr{\'{a}}lovic.
\newblock Information complexity of online problems.
\newblock In {\em 35th International Symposium on Mathematical Foundations of
  Computer Science (MFCS) 2010}, volume 6281 of {\em Lecture Notes in Computer
  Science}, pages 24--36. Springer, 2010.

\bibitem{JansenS03}
Klaus Jansen and Roberto Solis{-}Oba.
\newblock An asymptotic fully polynomial time approximation scheme for bin
  covering.
\newblock {\em Theoretical Computer Science}, 306(1--3):543--551, 2003.

\bibitem{KarmarkarK82}
Narendra Karmarkar and Richard~M. Karp.
\newblock An efficient approximation scheme for the one-dimensional bin-packing
  problem.
\newblock In {\em 23rd Annual Symposium on Foundations of Computer Science
  (FOCS)}, pages 312--320. {IEEE} Computer Society, 1982.

\bibitem{Komm16}
Dennis Komm.
\newblock {\em An Introduction to Online Computation -- Determinism,
  Randomization, Advice}.
\newblock Texts in Theoretical Computer Science. An {EATCS} Series. Springer,
  2016.

\bibitem{Mikkelsen16}
Jesper~W. Mikkelsen.
\newblock Randomization can be as helpful as a glimpse of the future in online
  computation.
\newblock In {\em 43rd International Colloquium on Automata, Languages, and
  Programming (ICALP)}, pages 39:1--39:14. Springer, 2016.

\bibitem{Renault17}
Marc~P. Renault.
\newblock Online algorithms with advice for the dual bin packing problem.
\newblock {\em Central European Journal of Operations Research},
  25(4):953--966, 2017.

\bibitem{RenaultRS15}
Marc~P. Renault, Adi Ros{\'{e}}n, and Rob van Stee.
\newblock Online algorithms with advice for bin packing and scheduling
  problems.
\newblock {\em Theoretical Computer Science}, 600:155--170, 2015.

\bibitem{R13}
Thomas Rothvoss.
\newblock Approximating bin packing within o(log {OPT} * log log {OPT}) bins.
\newblock In {\em 54th Annual {IEEE} Symposium on Foundations of Computer
  Science (FOCS)}, pages 20--29. {IEEE} Computer Society, 2013.

\end{thebibliography}

\end{document}